\documentclass[twocolumn,journal]{IEEEtran}

\usepackage{amsfonts}
\usepackage{amsmath}
\usepackage{amsthm}
\usepackage{amssymb}
\usepackage{graphicx}
\usepackage[T1]{fontenc}
\usepackage{supertabular}
\usepackage{longtable}
\usepackage[usenames,dvipsnames]{color}
\usepackage{bbm}
\usepackage{fancyhdr}
\usepackage{breqn}
\usepackage{fixltx2e}
\usepackage{capt-of}
\usepackage{tikz}
\usetikzlibrary{matrix}
\usepackage{endnotes}
\usepackage{soul}
\usepackage{marginnote}

\newcommand{\mathsym}[1]{}
\newcommand{\unicode}[1]{}

\usepackage{float}

\usepackage[framemethod=TikZ]{mdframed}
\usepackage[framemethod=TikZ]{mdframed}
\usepackage{framed}

\mdfsetup{%
	outerlinewidth=1,skipabove=20pt,backgroundcolor=yellow!50, outerlinecolor=black,innertopmargin=0pt,splittopskip=\topskip,skipbelow=\baselineskip, skipabove=\baselineskip,ntheorem,roundcorner=5pt}

\mdtheorem[nobreak=true,outerlinewidth=1,
backgroundcolor=yellow!50, outerlinecolor=black,innertopmargin=0pt,splittopskip=\topskip,skipbelow=\baselineskip, skipabove=\baselineskip,ntheorem,roundcorner=5pt,font=\itshape]{result}{Result}

\mdtheorem[nobreak=true,outerlinewidth=1,
backgroundcolor=yellow!50, outerlinecolor=black,innertopmargin=0pt,splittopskip=\topskip,skipbelow=\baselineskip, skipabove=\baselineskip,ntheorem,roundcorner=5pt,font=\itshape]{theorem}{Theorem}

\mdtheorem[nobreak=true,outerlinewidth=1,
backgroundcolor=gray!10, outerlinecolor=black,innertopmargin=0pt,splittopskip=\topskip,skipbelow=\baselineskip, skipabove=\baselineskip,ntheorem,roundcorner=5pt,font=\itshape]{remark}{Remark}

\mdtheorem[nobreak=true,outerlinewidth=1,
backgroundcolor=pink!30, outerlinecolor=black,innertopmargin=0pt,splittopskip=\topskip,skipbelow=\baselineskip, skipabove=\baselineskip,ntheorem,roundcorner=5pt,font=\itshape]{quaestio}{Quaestio}

\mdtheorem[nobreak=true,outerlinewidth=1,
backgroundcolor=yellow!50, outerlinecolor=black,innertopmargin=5pt,splittopskip=\topskip,skipbelow=\baselineskip, skipabove=\baselineskip,ntheorem,roundcorner=5pt,font=\itshape]{background}{Background}

\usepackage{hyperref}
\begin{document}

\title{{\color{Red}
Tail Option Pricing Under Power Laws}}

\author{Nassim Nicholas Taleb\IEEEauthorrefmark{1}\IEEEauthorrefmark{2}\IEEEauthorrefmark{3}, Brandon Yarckin\IEEEauthorrefmark{1},  Chitpuneet Mann\IEEEauthorrefmark{1}, Damir Delic\IEEEauthorrefmark{1}, and Mark Spitznagel\IEEEauthorrefmark{1} 
    
   \IEEEauthorblockA{  \IEEEauthorrefmark{1} Universa Investments
   \IEEEauthorrefmark{2}Tandon School of Engineering, New York University\\   \IEEEauthorrefmark{3}Corresponding author, nnt1@nyu.edu \\
  }

Revision, March 2023}

\maketitle

\begin{mdframed}
\smallskip
\begin{abstract}

	We build a methodology that takes a given option price in the tails with strike $K$ and extends (for calls, all strikes > $K$, for puts all  strikes $< K$) assuming the continuation falls into what we define as "Karamata Constant" over which the strong Pareto law holds. The heuristic produces relative prices for options, with for  sole parameter the tail index $\alpha$, under some mild arbitrage constraints. 
	
	Usual restrictions such as finiteness of variance are not required.
	
	The heuristic allows us to scrutinize the volatility surface and test various theories of relative tail option overpricing (usually built on thin tailed models and minor modifications/fudging of the Black-Scholes formula).
\end{abstract}
\end{mdframed}
\begin{figure}
\includegraphics[width=\linewidth]{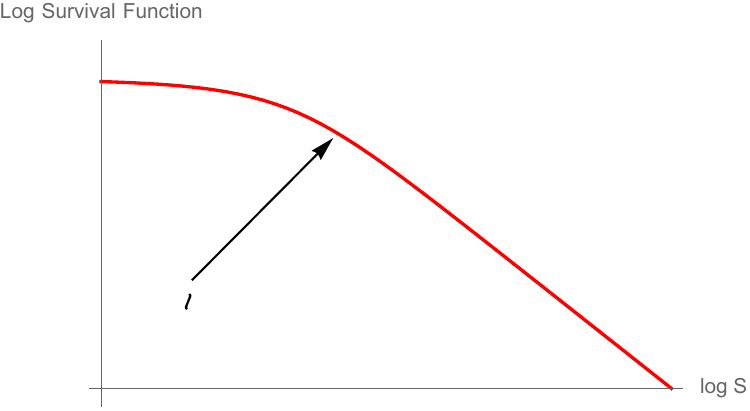}
 \caption{The Karamata constant where the slowly moving function is safely replaced by a constant $L(S) = l$. The constant varies whether we use the price S or its geometric return --but not the asymptotic slope which corresponds to the tail index $\alpha$.}\label{karamataconstant}
\end{figure}

\section{Introduction}
We\footnote{The introduction of power laws in option pricing took place in 1994 with the works of Bouchad and Sornette \cite{bouchaud1994black}; the main authors of this paper started implementing in their option trading the current heuristic methods in 2004 upon a meeting with Benoit Mandelbrot. Recent work along these lines was done by Hamidieh \cite{hamidieh2017estimating} for the opposite problem: estimating tail exponent from option prices.}
start by restating the conventional definition of the power law class, 
by the property of the survival function. Let $X$ be a random variable belonging to the class of distributions with a
"power law" (right) tail, that is: 
\begin{equation}
\mathbb{P}(X>x)\sim L(x)\,x^{-\alpha }  \label{powerlaweq}
\end{equation}%
where $L:\left[ x_{\min },+\infty \right) \rightarrow \left( 0,+\infty
\right) $ is a slowly varying function\index{Slowly varying function}, defined as $\lim_{x\rightarrow
+\infty }\frac{L(kx)}{L(x)}=1$ for any $k>0$, \cite{bingham1989regular}. 

The survival function of $X$ is called to belong to the "regular variation" class $RV_\alpha$. More specifically, a function $f:\mathbb{R}^+ \rightarrow \mathbb{R}^+$ is index varying at infinity  with index $\rho$ ($f \in RV_\rho$) when $$ \lim_{t\to \infty} \frac{f(t x)}{f(t)}=x^{\rho}$$.

More practically, there is a point where $L(x)$ approaches its limit
, becoming a constant as in Fig. \ref{karamataconstant}--we call it the "Karamata constant". Beyond such value the tails for power laws are calibrated using such standard techniques as the Hill estimator. The distribution in that zone is dubbed the strong Pareto law by B. Mandelbrot \cite{mandelbrot1960pareto},\cite{dyer1981structural}.\footnote{For risk neutral estimation of tail densities, see the reviews in Figlewski \cite{Figlewski2018Risk} and Reinke \cite{reinke2020risk}. Traditional methods of interpolation, one can see, are based on the Black-Scholes equation which fudges a Gaussian by changing the parameters per realization of the state variable --we can already note that our approach, thanks to the properties of power laws, is based on extrapolation. For an exposition of strike densities that is Black-Scholes free, see \cite{Taleb2015unique}.}

\section{Call Pricing beyond the "Karamata constant"}
\begin{figure}
\includegraphics[width=\linewidth]{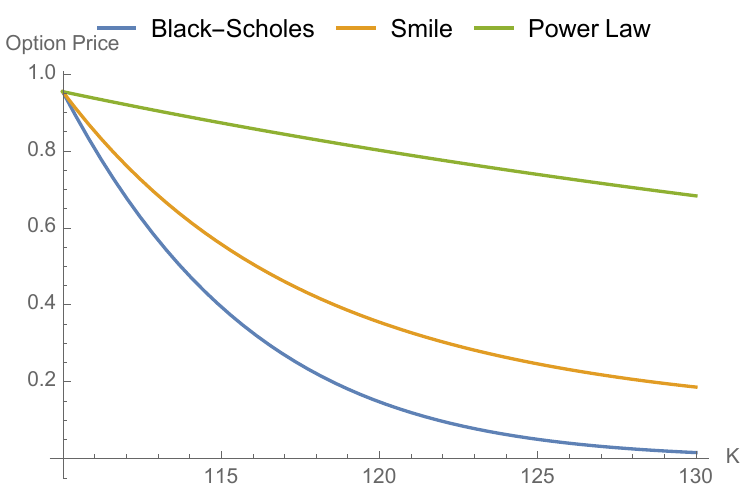}
\caption{We show a straight Black Scholes option price (constant volatility), one with a volatility "smile", i.e. the scale increases in the tails, and power law option prices. Under the simplified case of a power law distribution for the underlying, option prices are linear to strike.}\label{slope}
\end{figure}

Now define a European call price $C(K)$ with a strike $K$ and an underlying security price $S$, with $K, S \in \left( 0,+\infty
\right) $, as $(S-K)^+$, with its valuation performed under some probability measure $\mathbb{P}$, thus allowing us to price the option as $\mathbbm{E}_P(S-K)^+=\int_K^\infty (S-K) dS$.  This allows us to immediately prove the following results under the two main approaches. 
\subsection{First approach; the underlying, $S$ is in the regular variation class}
We start with a simplified case, to build the intuition.
	Let $S$ have a survival function in the regular variation class $RV_\alpha$ as per \ref{powerlaweq}. Let $\nu$ be a scaling constant. For all $K>l>0$ and $\alpha>1$,
\begin{equation}
C(K)=	\frac{K^{1-\alpha } l^{\alpha }}{\alpha -1} \nu\label{option price}.
\end{equation}
A brief comment: we used $l^\alpha$ rather than $l$, and introduced another constant $\nu$ to make the pseudo-density function integrate to unity --as we will see, these two constants disappear from the final equations. Although burdensome in exposition, the apparently unnecessary exponent $\alpha$ for $l$ makes our distribution in beyond the Karamata constant similar to the standard Pareto.
\begin{remark}
We note that the parameters $l$ and $\nu$, when derived from an existing option price, contains all necessary information about the probability distribution below $S=l$, which under a given $\alpha$ parameter makes it unnecessary to estimate the mean, the "volatility" (that is, scale) and other attributes.
\end{remark}

Let us assume that $\alpha$ is exogenously set (derived from fitting distributions, or, simply from experience, in both cases $\alpha$ is supposed to fluctuate minimally \cite{taleb2019statistical} ).
We note that $C(K)$ is invariant to distribution calibrations and the only parameters needed $l$ which, being constant, disappears in ratios.  Now consider as set the market price of an "anchor" tail option in the market is $C_m$ with strike $K_1$, defined as an option for the strike of which other options are priced in relative value.  We can simply generate all further strikes from $l=\left((\alpha -1) C_m K_1^{\alpha -1}\right)^{1/\alpha }$ and applying Eq. \ref{option price}. 
\begin{result}[Relative Pricing under Distribution for $S$]
	For $K_1, K_2 \geq l$, 
	\begin{equation}
			C(K_2)=\left(\frac{K_2}{K_1}\right)^{1-\alpha} C(K_1).\label{simplified ratio}
	\end{equation}
\end{result}
The advantage is that all parameters in the distributions are eliminated: all we need is the price of the tail option and the $\alpha$ to build a unique pricing mechanism.

\begin{remark}[Avoiding confusion about $L$ and $\alpha$]
The tail index $\alpha$ and Karamata constant $l$ should correspond to the assigned distribution for the specific underlying.	 A tail index $\alpha$ for $S $ in the regular variation class as as per \ref{powerlaweq} leading to Eq. \ref{option price} is different from that for $r=\frac{S-S_0}{S_0} \in RV_\alpha$ . For consistency, each should have its own Zipf plot and other representations. 
\begin{enumerate}
	\item If $\mathbb{P}(X>x)\sim L_a(x)\,x^{-\alpha }$, and $\mathbb{P}(\frac{X-X_0}{X_0}>\frac{x-X_0}{X_0})\sim L_b(x)\,x^{-\alpha }$, the $\alpha$ constant will be the same, but the the various $L_{(.)}$ will be reaching their constant level at a different rate.
	\item If $r_c=\log\frac {S}{S_0}$, it is not in the regular variation class, see theorem next.

\end{enumerate}
\end{remark}
The reason $\alpha$ stays the same is owing to the scale-free attribute of the tail index.
\begin{theorem}[Log returns]
	Let S be a random variable with survival function  $\varphi(s)=L(s) s^{-\alpha} \in RV_\alpha$, where $L(.)$ is a slowly varying function. Let $r_l$ be the log return 
	$r_l=\log\frac {s}{s_0} $. 
	$\varphi_{r_l}(r_l)$ is not in the $RV_\alpha$ class.
\end{theorem}
\begin{proof} Immediate, thanks to the transformation $ \varphi_{r_l}(r_l)=L(s) s^{-\frac{\log \left(\log ^{\alpha }(s)\right)}{\log (s)}}$.\end{proof}
We note, however, that in practice, although we may need continuous compounding to build dynamics \cite{taleb2009finiteness}, our approach assumes such dynamics are contained in the anchor option price selected for the analysis (or $l$). Furthermore there is no tangible difference, outside the far tail, between $\log\frac{S}{S_0}$ and $\frac{S-S_0}{S_0}$.

\subsection{Second approach, $S$ has geometric returns in the regular variation class}
Let us now apply to real world cases where the returns $\frac{S-S_0}{S_0}$ are Paretan. Consider, for $r>l$, $S= (1+r) S_0$, where $S_0$ is the initial value of the underlying and $r \sim \mathcal{P}(l,\alpha)$ (Pareto I distribution) with survival function 
\begin{equation}
	\left(\frac{K-S_0}{l S_0}\right)^{-\alpha},  \;  K> S_0 (1+l)
\end{equation}
and fit to $C_m$ using $l=\frac{(\alpha -1)^{1/\alpha } C_m^{1/\alpha } \left(K-S_0\right){}^{1-\frac{1}{\alpha }}}{S_0}$, which, as before shows that practically all information about the distribution is embedded in $l$.
 	
 Let $\frac{S-S_0}{S_0}$ be in the regular variation class. For $S \geq S_0(1+l)$,
\begin{equation}
	C(K,S_0)=\frac{(l\; S_0)^{\alpha }  (K-S_0)^{1-\alpha }}{\alpha -1}
\end{equation}

\begin{figure*}[h!]
\includegraphics[width=\linewidth]{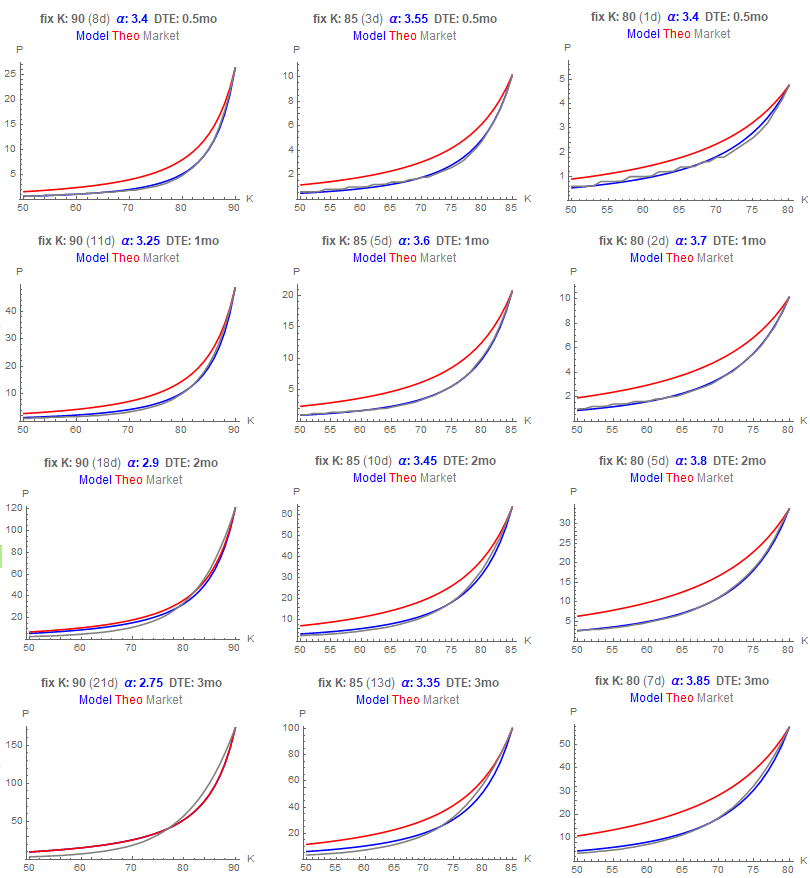}	\caption{Put Prices in the SP500 using "fix K" as anchor (from Dec 31, 2018 settlement), and generating an option prices using a tail index $\alpha$ that matches the market (blue) ("model), and in red prices for $\alpha= 2.75$. We can see that market prices tend to 1) fit a power law (matches stochastic volatility with fudged parameters), 2) but with an $\alpha$ that thins the tails. This shows how models claiming overpricing of tails are grossly misspecified.}\label{panel1}
\end{figure*}

\begin{figure*}[h!]
\includegraphics[width=\linewidth]{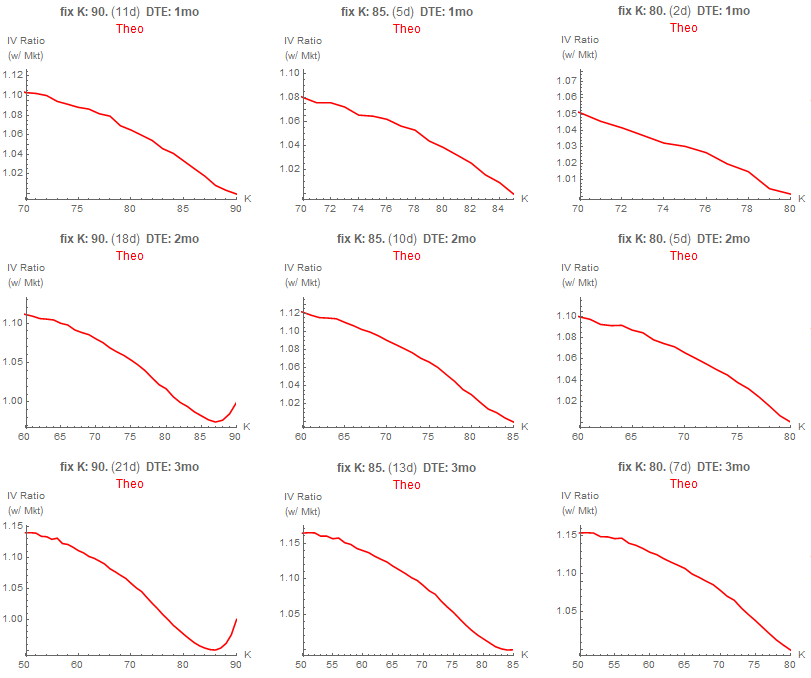}	\caption{Same results as in Fig \ref{panel1} but expressed using implied volatility. We match the price to implied volatility for downside strikes (anchor $90$, $85$, and $80$) using our model vs market, in ratios. We assume $\alpha=2.75$.}
\end{figure*}
We can thus rewrite Eq. \ref{simplified ratio} to eliminate $l$: 

\begin{result}[Relative Pricing under Distribution for $\frac{S-S_0}{S_0}$]
	
For  $K_1, K_2 \geq (1+l) S_0$,
\begin{equation}
	C(K_2)=\left(\frac{K_2-S_0}{K_1-S_0}\right)^{1-\alpha} C(K_1) .\label{compl ratio}
\end{equation}
\end{result}

\begin{figure}
\includegraphics[width=\linewidth]{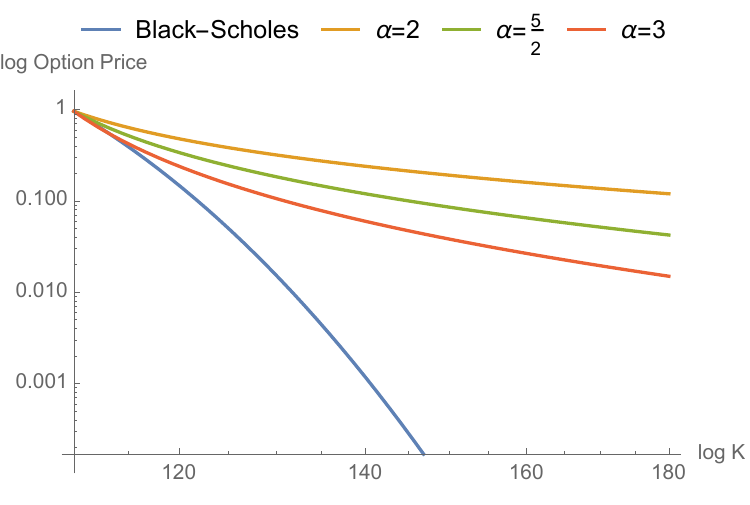}
\caption{The intuition of the Log log plot for the second calibration}\label{slope}
\end{figure}

Note: Unlike the pricing methods in the Black-Scholes modification class (stochastic and local volatility models, (see the expositions of Dupire \cite{dupire1994pricing}, Derman et al.\cite{demeterfi1999guide}, and Gatheral,  \cite{gatheral2006volatility},  finiteness of variance is not required neither for our model nor for option pricing in general, as shown in  \cite{taleb2009finiteness}. The only requirement is $\alpha>1$, that is, finite first moment.

\section{Put Pricing}
We now consider the put strikes (or the corresponding calls in the left tail, which should be priced via put-call parity arbitrage).  Unlike with calls, we can only consider the variations of $\frac{S-S_0}{S_0}$, not the logarithmic returns (nor those of $S$ taken separately).

We construct the negative side with a negative return for the underlying. Let $r$ be the rate of return $S= (1-r) S_0$, and Let $r>l>0$ be Pareto distributed in the positive domain, with density $f_r(r)= \alpha \; l^{\alpha } r^{-\alpha -1}$. We have by probabilistic transformation and rescaling the PDF of the underlying:

$$\begin{array}{cc}
  & 
\begin{array}{cc}
 f_S(S)=-\frac{\alpha  \left(-\frac{S-S_0}{l S_0}\right)^{-\alpha -1}}{l S_0} \lambda  & S \in \left[0,  (1-l)S_0\right) \\
\end{array}
 \\
\end{array}$$
where the scaling constant $\lambda=\left(\frac{1}{(-1)^{\alpha +1} \left(l^{\alpha }-1\right)}\right)$ is set in a way to get $f_s(S)$ to integrate to 1. The parameter $\lambda$, however, is close to $1$, making the correction negligible, in applications where $\sigma \sqrt{t} \leq \frac{1}{2}$ ($\sigma$ being the Black-Scholes equivalent implied volatility and $t$ the time to expiration of the option).

Remarkably, both the parameters $l$ and the scaling $\lambda$ are eliminated.
\begin{result}[Put Pricing]

	For  $K_1, K_2 \leq (1-l) S_0$,
	\begin{dmath}
P\left(K_{2}\right)=P\left(K_{1}\right) \frac{\left(K_{2}-S_{0}\right)^{1-\alpha}-S_{0}^{1-\alpha}\left((\alpha-1) K_{2}+S_{0}\right)}{\left(K_{1}-S_{0}\right)^{1-\alpha}-S_{0}^{1-\alpha}\left((\alpha-1) K_{1}+S_{0}\right)}	
\end{dmath}
\end{result}

\section{Arbitrage Boundaries}
Obviously, there is no arbitrage for strikes higher than the baseline one $K_1$ in previous equations. For we can verify the Breeden-Litzenberger result \cite{breeden1978prices}, where the density is recovered from the second derivative of the option with respect to the strike $\frac{\partial ^2 C(K)}{\partial K^2}|_{K\geq K_1}=\alpha  K^{-\alpha -1} L^{\alpha } \geq 0$. 

However there remains the possibility of arbitrage between strikes $K_1+\Delta K$, $K_1$, and $K_1-\Delta K$ by violating the following boundary: let $BSC(K,\sigma(K))$ be the Black-Scholes value of the call for strike $K$ with volatility $\sigma(K)$ a function of the strike and $t$ time to expiration. We have 
\begin{equation}
	C(K_1+\Delta K)+BSC(K_1-\Delta K) \geq 2\; C(K_1)\label{arbitrage},
\end{equation}
where $BSC(K_1,\sigma(K_1))=C(K_1)$. For inequality \ref{arbitrage} to be satisfied, we further need an inequality of call spreads, taken to the limit:
\begin{equation}
	\frac{\partial BSC(K,\sigma(K) )}{\partial K}|_{K=K_1}\geq \frac{\partial C(K )}{\partial K}|_{K=K_1}
\end{equation}

Such an arbitrage puts a lower bound on the tail index $\alpha$. Assuming 0 rates to simplify:
\begin{footnotesize}
\begin{dmath}
\alpha \geq \frac{1}{-\log \left(K-S_0\right)+\log (l)+\log
   \left(S_0\right)}\\
   	\log \left(\frac{1}{2} \text{erfc}\left(\frac{t \sigma (K)^2+2 \log
   (K)   -2 \log \left(S_0\right)}{2 \sqrt{2} \sqrt{t} \sigma (K)}\right)-\frac{\sqrt{S_0}
   \sqrt{t} \sigma '(K) K^{\frac{\log \left(S_0\right)}{t \sigma (K)^2}+\frac{1}{2}}
   \exp \left(-\frac{\log ^2(K)+\log ^2\left(S_0\right)}{2 t \sigma (K)^2}-\frac{1}{8} t
   \sigma (K)^2\right)}{\sqrt{2 \pi }}\right)
\end{dmath}
   \end{footnotesize}

\section{Comments}
As we can see in Fig. \ref{slope}, stochastic volatility models and similar adaptations (say, jump-diffusion or standard Poisson variations) eventually fail "out in the tails" outside the zone for which they were calibrated. There has been poor attempts to extrapolate the option prices using a fudged thin-tailed probability distribution rather than a Paretan one --hence the numerous claims in the finance literature on "overpricing" of tail options combined with some psychological comments on "dread risk" are unrigorous on that basis. The proposed methods allows us to approach such claims with more realism. 

Finaly, note that our approach isn't about absolute mispricing of tail options, but relative to a given strike closer to the money.

\section*{Acknowledgments} 
Bruno Dupire, Peter Carr, students at NYU Tandon School of Engineering, Bert Zwart and the participants at the Heavy Tails Workshop (April 6-9 2019) in Eindhoven he helped organize. We are particularly thankful to Joe Pimbley for invaluable help.

\bibliographystyle{IEEEtran}
\bibliography{/Users/nntaleb/Dropbox/Central-bibliography}

\end{document}